\documentclass[conference]{IEEEtran}
\usepackage{amsmath}
\usepackage{amsthm}
\usepackage{amssymb}
\usepackage{graphicx}
\usepackage{subcaption}
\usepackage{url}
\usepackage{calrsfs}
\usepackage{amsfonts}
\usepackage{tikz}
\usepackage[para,online,flushleft]{threeparttable}
\newcommand{\mytilde}{\raise.17ex\hbox{$\scriptstyle\mathtt{‌​\sim}$}}
\begin{document}

\title{Wireless Bidirectional Relaying using Physical Layer Network Coding with Heterogeneous PSK Modulation}

\author{\IEEEauthorblockN{Chinmayananda A.}
\IEEEauthorblockA{Dept. of Electrical Communication Engg.\\ Indian Institute of Science Bangalore, \\Karnataka (India) - 560012\\
Email: chinmaya@ece.iisc.ernet.in}
\and
\IEEEauthorblockN{Saket D. Buch}
\IEEEauthorblockA{Dept. of Electrical Engg.\\
Indian Institute of Science Bangalore, \\Karnataka (India) - 560012\\
Email: saket@ee.iisc.ernet.in}
\and
\IEEEauthorblockN{B. Sundar Rajan}
\IEEEauthorblockA{Dept. of Electrical Communication Engg.\\ Indian Institute of Science Bangalore, \\Karnataka (India) - 560012\\
Email: bsrajan@ece.iisc.ernet.in}}

\newtheorem{theorem}{Theorem}
\newtheorem{definition}{Definition}
\newtheorem{lemma}{Lemma}
\newtheorem{example}{Example}

\maketitle

\begin{abstract}
In bidirectional relaying using Physical Layer Network Coding (PLNC), it is generally assumed that users employ same modulation schemes in the Multiple Access phase. However, as observed by Zhang et al.  \cite{zhang2016design}, 
it may not be desirable for the users to always use the same modulation schemes, particularly when user-relay channels are not equally strong. Such a scheme is called Heterogeneous PLNC. However, the approach in \cite{zhang2016design} uses the computationally intensive Closest Neighbour Clustering (CNC) algorithm to find the network coding maps to be applied at the relay. Also, the treatment is specific to certain cases of heterogeneous modulations. In this paper, we show that, when users employ heterogeneous \emph{symmetric}-PSK modulations, the network coding maps and the mapping regions in the \emph{fade state} plane can be obtained analytically. Performance results are provided in terms of Relay Error Rate (RER) and Bit Error Rate (BER).  

\end{abstract}
\IEEEpeerreviewmaketitle
\section{Introduction}
 A bidirectional or two-way relaying scenario consists of nodes A and B wanting to exchange data using a relay R. If the relay performs PLNC, the relaying has two phases: Multiple Access (MA) phase and Broadcast (BC) phase \cite{koike2009optimized}. In the MA phase, users transmit signals simultaneously to the relay. The signal received at the relay is a noisy sum of transmitted signals scaled by their respective channel coefficients. The relay applies a many-to-one map on the received symbol, such that the users can decode the desired message, given their own message and a knowledge of the map. The salient feature of this scheme is that the mapping depends on the \emph{fade  state} of user-relay channels. Hence, the complex plane representing the ratio of channel coefficients (or fade state) has to be partitioned to indicate which map is to be used in a given region. Relaying carried out in this manner is highly efficient compared to conventional relaying and Network Coding at bit-level. \cite{koike2009optimized}.

 Both users are assumed to apply the same constellation during the MA phase \cite{koike2009optimized,muralidharan2013wireless,namboodiri2013physical,zhang2006hot}. However, due to varying channel conditions (i.e. average SNR), each link may only be able to support constellations up to a certain cardinality. Such a scheme is known as Heterogeneous PLNC (HePNC) which was proposed by Zhang et al. \cite{zhang2014hepnc, zhang2016design}. Three cases of HePNC viz. QPSK-BPSK, 8PSK-BPSK, and 16QAM-BPSK are investigated in \cite{zhang2016design} and the network coding maps are obtained using CNC algorithm \cite{koike2009optimized}. However, a computationally efficient way of obtaining the network coding maps and the corresponding fade state region boundaries using the concept of \emph{singular fade states} (SFS), is known \cite{muralidharan2013wireless} for \emph{symmetric} PSK modulations. Also, the concept of \emph{clustering independent} regions is not considered in \cite{zhang2014hepnc,zhang2016design}.
 
 In this paper, we provide an analytical approach to the PLNC scheme with heterogeneous PSK modulations by extending the framework of Muralidharan et al.\cite{muralidharan2013wireless}.
The contributions of this paper are:
 \begin{itemize} 
 	\item The analytical framework provided in this paper enables the investigation of two-way relaying using HePNC with PSK modulations. This includes cases like 8PSK-QPSK PLNC which are not considered by Zhang et al.
 	\item The approach taken by Zhang et al. uses a computer search based CNC algorithm \cite{koike2009optimized} to find the boundaries of \emph{relay-mapping regions}. The complexity increases with the order of PSK signal sets used. This work provides explicit equations for the boundaries of \emph{relay-mapping} regions.
 	\item For the HePNC modulations dealt by Zhang et al, where one of the users employs BPSK signal set, the boundaries of relay-mapping consist of straight lines only. For a general case, the boundaries also include arcs of circles, for which this paper gives explicit equations. 
 	\item This work generalizes some results of \cite{muralidharan2013wireless}. For example, if same order of PSK modulation is used at both users, the \emph{Internal} Clustering Independent (CI) Region is obtained from complex inversion of \emph{External} CI Region. This paper provides equations to obtain boundaries of Internal CI regions for heterogeneous PSK modulations. 
 \end{itemize}
 The paper is organized as follows: Section II describes the system model, which is followed by the derivation of the exact number and location of SFS in Section III. The Latin rectangles used for QPSK-BPSK and 8PSK-BPSK schemes, which serve as many-to-one maps at the relay are listed in Section IV. The analytical framework for partitioning of the fade state plane is given in Section V. Thereafter, the performance is evaluated in terms of Bit and Relay Error Rates in AWGN and Rayleigh fading channels. The results match those provided by Zhang et al. \cite{zhang2016design}. 
\section{System Model}
The HePNC system comprises of two users A and B, wanting to exchange data through a relay R. The notation from \cite{muralidharan2013wireless} and \cite{namboodiri2013physical} has been used and extended. Let A and B use PSK modulations $\mathcal{S}_i$ of cardinality $M_i = 2^{\lambda_i}$,  where, $\lambda_i \in \mathbb{Z}^+$ for $i=1,2$. Subscripts $i=1, 2$ is used for A, B interchangeably, to indicate parameters of A and B respectively. That is, A wants to send a $\lambda_1$ binary-bit tuple to B, and B wants to send a $\lambda_2$ binary-bit tuple to A. Without loss of generality, $\lambda_1 \geq \lambda_2$. Let the functions mapping the bit-tuples to complex symbols be binary to decimal mappings given by $\mu_i: \mathbb{F}_2^{\lambda_i} \rightarrow \mathcal{S}_i$ for $i=1, 2$. This paper considers users using symmetric $M_i$-PSK constellations of the form $\mathcal{S}_i = \left\lbrace e^{\frac{j(2k+1)\pi}{M_i}}; 0 \leq k \leq M_i-1 \right\rbrace$. The relaying has two phases, MA phase and BC phase. 
\begin{figure}
\centering
\begin{tikzpicture}[scale=1.5]
code
\node[circle, minimum size=1cm, fill=gray!30] (a) at (0, 3) {\Large{A}};
\node[circle, minimum size=1cm, fill=gray!30] (b) at +(1.5,3) {\Large{R}};
\node[circle, minimum size=1cm, fill=gray!30] (c) at +(3,3) {\Large{B}};
\draw [->, thick] (a) -- (b) node[pos=.5,sloped,above] {$H_{A}$};
\draw [->, thick] (c) -- (b) node[pos=.5,sloped,above] {$H_{B}$};
\node at (0,2.5) {$x_A \in \mathcal{S}_1$};
\node at (3,2.5) {$x_B \in \mathcal{S}_2$};
\node at (1.5,2) {\small{(a) Multiple Access Phase  [1 Channel Use]}};

\node[circle, minimum size=1cm, fill=gray!30] (a) at (0, 1) {\Large{A}};
\node[circle, minimum size=1cm, fill=gray!30] (b) at +(1.5,1) {\Large{R}};
\node[circle, minimum size=1cm, fill=gray!30] (c) at +(3,1) {\Large{B}};
\draw [->, thick] (b) -- (a) node[pos=.5,sloped,above] {$H'_{A,j}$};
\draw [->, thick] (b) -- (c) node[pos=.5,sloped,above] {$H'_{B,j}$};
\node at (1.5,0.5) {$x_{R,j} \in \mathcal{S}_2$, $j \in {1,2,...,N_{t}}$};
\node at (1.5,0) {\small{(b) Broadcast Phase  [$N_{t}=\lceil \frac{log_2|\mathcal{S}_R|}{\lambda_2} \rceil$ Channel Uses]}};
\end{tikzpicture}
\caption{Two-way Relaying with Heterogeneous PLNC \cite{zhang2016design}}
\end{figure}
\subsubsection{MA Phase}
In this phase, the users transmit complex symbols to the relay simultaneously. The user transmissions, $x_A = \mu_1(s_A)$ and $x_B = \mu_2(s_B)$, where, $s_A \in \mathbb{F}_2^{\lambda_1}$ and $s_B \in \mathbb{F}_2^{\lambda_2}$. The received signal at R is given by 
\begin{equation}
Y_R = H_Ax_A + H_Bx_B + Z_R,  
\end{equation}
where $H_A$ and $H_B$ are fade coefficients of the A-R and B-R links respectively. The additive noise $Z_R$ is $\mathcal{C}\mathcal{N}(0,\sigma^2)$, where $ \mathcal{CN}(0,\sigma^2)$ denotes circularly symmetric complex Gaussian random variable with zero mean and variance $\sigma^2$. Slow fading is assumed. The pair $(\gamma,\theta)$ is defined as a fade state.where
\begin{equation}
\gamma e^{j\theta} \triangleq \frac{H_B}{H_A}; \gamma \in \mathbb{R}^+, -\pi \leq \theta \leq \pi.
\end{equation}
It is assumed that the fade states are distributed according to a continuous probability distribution. The channel state information is assumed to be available only at the receivers. Thus, no carrier phase synchronization is needed at the users. However, symbol-level timing synchronization is assumed \cite{zhang2016design}. The effective constellation seen at the relay during the MA phase is therefore, 
\begin{equation}
\mathcal{S}_R(\gamma,\theta) = \left\lbrace x_A + \gamma e^{j\theta} x_B| x_A \in \mathcal{S}_1, x_B \in \mathcal{S}_2 \right\rbrace.
\end{equation}
The minimum distance between any two points in the constellation $\mathcal{S}_R(\gamma,\theta)$ is defined as,
\begin{equation}
d_{min}(\gamma,\theta) = \underset{\underset{(x_A,x_B)\neq(x_A',x_B')}{(x_A,x_B),(x_A',x_B') \in \mathcal{S}_1 \times \mathcal{S}_2}}{min} |(x_A-x_A') + \gamma e^{j\theta}(x_B-x_B')|.
\end{equation}
A fade state where $d_{min}(\gamma,\theta) = 0$ is called a \emph{Singular Fade State (SFS)}. The set of all SFS, is $\mathcal{H} = \left\lbrace\gamma e^{j\theta} \in \mathbb{C}|d_{min}(\gamma,\theta) = 0\right\rbrace$. The relay R jointly decodes the transmitted complex symbol pair $(x_A,x_B)$ by computing a maximum likelihood estimate:
\begin{equation}
(\hat{x}_A,\hat{x}_B) = \underset{(x_A',x_B') \in \mathcal{S}_1 \times \mathcal{S}_2}{arg min}  |Y_R - H_Ax_A' - H_Bx_B'|.
\end{equation}
\subsubsection{BC Phase}
Instead of re-transmitting the estimated pair in the form of a constellation of cardinality $M_1 \times M_2$, the relay applies a many-to-one map from this pair to points on another constellation having smaller cardinality. This many-to-one map depends on the fade state and is defined by $\mathcal{M}^{\gamma,\theta}: \mathcal{S}_1 \times \mathcal{S}_2 \rightarrow \mathcal{S}'(\gamma,\theta)$. The signal set $\mathcal{S}'$ satisfies the inequality, $max(|\mathcal{S}_1|,|\mathcal{S}_2|) \leq |\mathcal{S}'(\gamma,\theta)| \leq |\mathcal{S}_1| \times |\mathcal{S}_2|$. Elements of $\mathcal{S}_1 \times \mathcal{S}_2$ mapped to the same complex number in $\mathcal{S}'$  by the map $\mathcal{M}^{\gamma,\theta}$ are said to form a cluster. Let $\left\lbrace\mathcal{L}_1, \mathcal{L}_2, ..., \mathcal{L}_l\right\rbrace$ denote the set of all clusters for a given fade state, also called a Clustering. Let $\mathcal{C}$ denote a generic clustering. Let $\mathcal{L}_{\mathcal{C}}(x_A,x_B)$ denote the cluster to which $(x_A,x_B)$ belongs under the clustering $\mathcal{C}$. 

The relay needs to broadcast a complex symbol $X_R = \mathcal{M}^{\gamma,\theta}(\hat{x}_A,\hat{x}_B) \in \mathcal{S}'(\gamma,\theta)$. The cardinality of the transmitted constellation is assumed to be limited by the SNR of the weakest link (by convention B-R). Thus, the relay transmits a set of symbols $X_{R_j} \in \mathcal{S}_2$ with the number of broadcast transmissions being $N_t = \lceil \frac{log_2(|\mathcal{S}'(\gamma,\theta)|)}{|\lambda_2|} \rceil$ and $\bar{X}_R = \left\lbrace X_{R_1}, X_{R_2}, ..., X_{R_{N_t}} \right\rbrace$. The users A and B receive transmissions $Y_{A_j} = H_{A_j}'x_{R_j} + Z_{A_j}$ and $Y_{B_j} = H_{B_j}'x_{R_j} + Z_{B_j}$, where, $j \in \{1,..., N_t\}$. The fading coefficients corresponding to the R-A and R-B links are denoted by $H_{A_j}'$ and $H_{B_j}'$ and the additive noises $Z_{A_j}$ and $Z_{B_j}$ are $\mathcal{C}\mathcal{N}(0,\sigma^2)$. The users then decode the individual transmissions and create composite symbols $Y_A$ and $Y_B$ to estimate $x_A$ and $x_B$ respectively. The map $\mathcal{M}^{\gamma,\theta}$ is known to the users, and so is the symbol transmitted by them, using which the data of the other user has to be recovered. 
To ensure this, the many-to-one map should satisfy the condition called Exclusive Law \cite{koike2009optimized,muralidharan2013wireless}, which is 
\begin{align}
\mathcal{M}^{\gamma,\theta}&(x_A,x_B) \neq \mathcal{M}^{\gamma,\theta}(x_A',x_B), \nonumber \\
&~for~ x_A \neq x_A'; x_A, x_A' \in \mathcal{S}_1, \forall x_B \in \mathcal{S}_2. \nonumber \\ 
\mathcal{M}^{\gamma,\theta}&(x_A,x_B) \neq \mathcal{M}^{\gamma,\theta}(x_A,x_B'), \nonumber\\
&~for~ x_B \neq x_B'; x_B, x_B' \in \mathcal{S}_2, \forall x_A \in \mathcal{S}_1.
\end{align}
This constraint leads to the mapping function being of the form of a \emph{Latin Rectangle}. If the fade state is an SFS, the relay cannot decide upon the transmitted pair $(x_A,x_B)$, as multiple pairs lead to the same received symbol at the relay. For fade state values $(\gamma,\theta)$ near the neighbourhood of an SFS, the value of $d_{min}(\gamma,\theta)$ is greatly reduced, which might lead the relay mapping the estimated transmitted symbols to a wrong constellation point in the MA phase. To mitigate this harmful effect of an SFS, another constraint called the Singularity Removal Constraint is imposed: for all pairs $(x_A,x_B)$ and $(x_A',x_B')$, where, $x_A \neq x_A'$ and $x_B \neq x_B'$, such that $|(x_A-x_A') + \gamma e^{j\theta}(x_B-x_B')| = 0$, ensure $\mathcal{M}^{\gamma,\theta}(x_A,x_B) = \mathcal{M}^{\gamma,\theta}(x_A',x_B')$. 
The above mappings remove the detrimental effect of \emph{distance shortening} \cite{koike2009optimized}. 
The minimum clustering distance for a given mapping/clustering $\mathcal{C}$ and a given fade state ($\gamma,\theta$) is defined as
\begin{align}
& d_{min}(\mathcal{C};\gamma,\theta) = \nonumber\\ & \underset{ \underset{\mathcal{M}^{\gamma,\theta}(x_A,x_B) \neq \mathcal{M}^{\gamma,\theta}(x_A',x_B')}{(x_A,x_B),(x_A',x_B') \in \mathcal{S}_1 \times \mathcal{S}_2}}{min} |(x_A-x_A') + \gamma e^{j\theta}(x_B-x_B')|.
\end{align}
A mapping is said to \emph{remove} an SFS $h \in \mathcal{H}$, if the minimum clustering distance $d_{min}(\mathcal{C};h) > 0$. 

\section{Singular Fade States}
In this section, the location and number of SFS are obtained for $M_1$-PSK--$M_2$-PSK PLNC. The points in the $M_i$-PSK signal set are assumed to be of the form  $\mathcal{S}_i = \left\lbrace e^{j(2k+1)\pi/M_i},0 \leq k \leq M_i-1 \right\rbrace$, where, $M_i = 2^{\lambda_i}$ and $\lambda_i \in \mathbb{Z}^+$ for $i$={1,2}. The \emph{difference constellation} $\Delta\mathcal{S}_i$, of a signal set $\mathcal{S}_i$, is given by
\begin{equation}
\Delta \mathcal{S}_i = \left\lbrace s_l - s_m | s_l,s_m \in \mathcal{S}_i \right\rbrace.
\end{equation} 
For a symmetric $M_i$-PSK signal set, we have,
\begin{align}
s_l - s_m &= e^{j(2l+1)\pi/M_i} - e^{j(2m+1)\pi/M_i} \nonumber \\
&= \left\lbrace cos\left(\frac{(2l+1)\pi}{M_i}\right) - cos\left(\frac{(2m+1)\pi}{M_i}\right) \right\rbrace \nonumber \\ &+ j\left\lbrace sin\left(\frac{(2l+1)\pi}{M_i}\right) - sin\left(\frac{(2m+1)\pi}{M_i}\right) \right\rbrace \nonumber \\
&= 2sin\left(\frac{(l-m)\pi}{M_i}\right) e^{j\left(\frac{(l+m+1)\pi}{M_i}\right)}.
\end{align}
Let $l-m = n$. It is sufficient to consider $n$ in the range $1 \leq n \leq M_i/2$ to get all the members of $\Delta\mathcal{S}_i$ . Let $l+m+1 = 2k$, if $n$ is odd, and $l+m+1 = 2k+1$, if $n$ is even, and  $0 \leq k \leq M_i-1$. Thus, we have
\begin{align}
\Delta\mathcal{S}_i = \left\lbrace0\right\rbrace &\bigcup \left\lbrace \underset{\underset{0 \leq k \leq M-1}{ 1 \leq n \leq M_i/2; n~odd}}{\bigcup} 2sin\left(\frac{n\pi}{M_i}\right) e^{j\left(\frac{2k\pi}{M_i}\right)} \right\rbrace\nonumber \\ &\bigcup \left\lbrace \underset{\underset{0 \leq k \leq M_i-1}{ 1 \leq n \leq M_i/2; n~even}}{\bigcup} 2sin\left(\frac{n\pi}{M_i}\right) e^{j\left(\frac{(2k+1)\pi}{M_i}\right)} \right\rbrace.\nonumber
\end{align}  
Thus, $x_{k,n,i} \in \Delta\mathcal{S}_i$, is given by,
\begin{equation}
x_{k,n,i} \triangleq \left\{ \,
\begin{IEEEeqnarraybox}[][c]{l?s}
\IEEEstrut
2sin\left(\frac{n\pi}{M_i}\right) e^{j\left(\frac{(2k+1)\pi}{M_i}\right)} & if $n~ even$, \\
2sin\left(\frac{n\pi}{M_i}\right) e^{j\left(\frac{2k\pi}{M_i}\right)} & if $n~ odd$,
\IEEEstrut
\end{IEEEeqnarraybox}
\right.
\label{diff_const_eqn}
\end{equation}
where $1 \leq n \leq M_i/2$, $0 \leq k \leq M_i-1$. Thus, SFS are of the form $\gamma_se^{j\theta_s} = \frac{H_B}{H_A} = -\frac{x_{k_1,n_1,1}}{x_{k_2,n_2,2}}$, for some $x_{k_i,n_i,i} \in \Delta\mathcal{S}_{i}$. Let $\delta \triangleq log_2(M_1/M_2)$.
\begin{lemma}[Muralidharan et al. \cite{muralidharan2013wireless}]
For integers $k_1,k_2,l_1$ and $l_2$, where $1 \leq k_1,k_2,l_1,l_2 \leq M/2$, $k_1 \neq k_2$ and $l_1 \neq l_2$,
\begin{equation}
\frac{sin\left(\frac{k_1\pi}{M}\right)}{sin\left(\frac{k_2\pi}{M}\right)} = \frac{sin\left(\frac{l_1\pi}{M}\right)}{sin\left(\frac{l_2\pi}{M}\right)},
\end{equation}
if and only if $k_1 = l_1$ and $k_2 = l_2$.
\end{lemma}
\begin{lemma}
The singular fade states lie on $\frac{M_1M_2}{4} - \frac{M_2}{2} + 1$ circles with $M_1$ points on each circle $(M_1 \geq M_2)$ with radii of circles given by $sin\left(n_1\pi/M\right)/sin\left(n_2\pi/M\right)$, where $1 \leq n_i \leq M_i/2$. The phase angles of the $M_1$ points are given by $2l\pi /M_1 + \phi$, where $0 \leq l \leq M_1-1$ and $\phi$ is,
\begin{equation}
\phi = \begin{cases}
	0 & \text{if $n_1$ is odd and $n_2$ is odd},\\
	\frac{(1-2^{\delta})\pi}{M_1} & \text{if $n_1$ is even and $n_2$ is even},\\
	\frac{-2^{\delta}\pi}{M_1} & \text{if $n_1$ is odd and $n_2$ is even},\\
	\frac{\pi}{M_1} & \text{if $n_1$ is even and $n_2$ is odd}.
	\end{cases}
\end{equation}
\end{lemma}
\begin{proof}
From (\ref{diff_const_eqn}), the amplitudes of SFS $\gamma_s=  sin\left(n_1\pi/M_1\right)/sin\left(n_2\pi/M_2\right)$ for some $n_i \in \left[M_i/2\right]$. 

We need to count the number of distinct values of $\gamma_s$. From Lemma 1 and using $M_1 = 2^{\delta}M_2$,
\begin{equation}
\frac{sin\left(\frac{l_1\pi}{M_1}\right)}{sin\left(\frac{2^{\delta}m_1\pi}{M_1}\right)} = \frac{sin\left(\frac{l_2\pi}{M_1}\right)}{sin\left(\frac{2^{\delta}m_2\pi}{M_1}\right)},
\end{equation}
where $1 \leq l_1,l_2,2^{\delta}m_1,2^{\delta}m_2 \leq M_1/2$, $l_1 \neq 2^{\delta}m_1$ and $l_2 \neq 2^{\delta}m_2$, if and only if $l_1 = l_2$ and $2^{\delta}m_1 = 2^{\delta}m_2 \Rightarrow m_1 = m_2$. 
Thus, out of the $M_1M_2/4$ pairs of $l_1$ and $m_1$, we subtract cases for which $l_1 = 2^{\delta}m_1$ (since they all lead to same $\gamma_s$ and add one on behalf of all of them. Hence, the number of distinct amplitudes of singular fade states is $\frac{M_1M_2}{4} - \frac{M_2}{2} + 1$.

From (\ref{diff_const_eqn}), the phase of singular fade states on the circles of different radii depend on the values of $n_i$ . If $n_1$ is odd and $n_2$ is also odd, the phase $\phi = 2k_1\pi/M_1 - 2k_22^{\delta}\pi/M_1$, where $0 \leq k_i \leq M_i -1$, $i = 1,2$. Taking $k_1-2^{\delta}k_2 = l$, it is clear that $l$ has $M_1$ distinct values and hence $0 \leq l \leq M_1 -1$.  Thus, in this case the phase of points is $\phi = 2\pi l/M_1$, which shows that there are $M_1$ equispaced SFS on each circle. 

The other cases follow similarly. If $n_1$ is odd and $n_2$ is even, the phase $\phi = 2k_1\pi/M_1 - 2k_22^{\delta}\pi/M_1 - 2^{\delta}\pi/M_1 = 2l \pi /M_1 - 2^{\delta}\pi/M_1$. If $n_1$ is even and $n_2$ is odd, the phase $\phi = 2k_1\pi/M_1 + \pi/M_1 - 2k_22^{\delta}\pi/M_1 = 2\pi l/M_1 + \pi/M_1$.Finally, if both $n_1$ and $n_2$ are even, $\phi = 2k_1\pi/M_1 + \pi/M_1 - 2k_22^{\delta}\pi/M_1 - 2^{\delta}\pi/M_1 = 2\pi l/M_1 + (1-2^{\delta})\pi/M_1$.
\end{proof}
When $\delta = 0$, we get the result of Lemma 2 of \cite{muralidharan2013wireless} as a special case. Hence, Lemma 2 is the generalized version of the corresponding result given in \cite{muralidharan2013wireless}.
\begin{example}
Let users A and B use QPSK and BPSK signal sets respectively. Thus, $x_A \in \left\lbrace \pm \frac{1}{\sqrt{2}} \pm \frac{j}{\sqrt{2}}\right\rbrace$ and $x_B \in \left\lbrace \pm j\right\rbrace$, $M_1=4$ and $M_2=2$. $\Delta\mathcal{S}_1 = \frac{1}{\sqrt{2}} \{0,2,-2,2j,-2j,-2+2j,-2-2j,2-2j,2+2j\}$ and $\Delta\mathcal{S}_2 = \{-2j,0,2j\}$. From the definition of SFS, we get,
\begin{equation*}
\mathcal{H}=\left\lbrace0,\frac{1}{\sqrt{2}},\frac{-1}{\sqrt{2}},\frac{j}{\sqrt{2}},\frac{-j}{\sqrt{2}},\frac{1+j}{\sqrt{2}},\frac{1-j}{\sqrt{2}},\frac{-1+j}{\sqrt{2}},\frac{-1-j}{\sqrt{2}} \right\rbrace
\end{equation*}
If we use Lemma 2, we see that the SFS are distributed in $4\times2/4 - 2/2 +1 = 2$ circles. The radii of the circles can be computed by taking cases of $n_1=1,~n_2=1$ and $n_1=2,~n_2=1$, so that $\gamma_s = \{1/\sqrt{2},1\}$. There are $M_1 = 4$ points on each circle with phases, as given in Lemma 2.  
\end{example}
The SFS for QPSK-BPSK, 8PSK-BPSK, and 8PSK-QPSK are shown in Fig. \ref{sfs:sub1}, \ref{sfs:sub2}, and \ref{sfs:sub3} respectively.
\begin{figure*}[!htbp]
\centering
\begin{subfigure}{.33\textwidth}
  \centering
  \includegraphics[width=15pc]{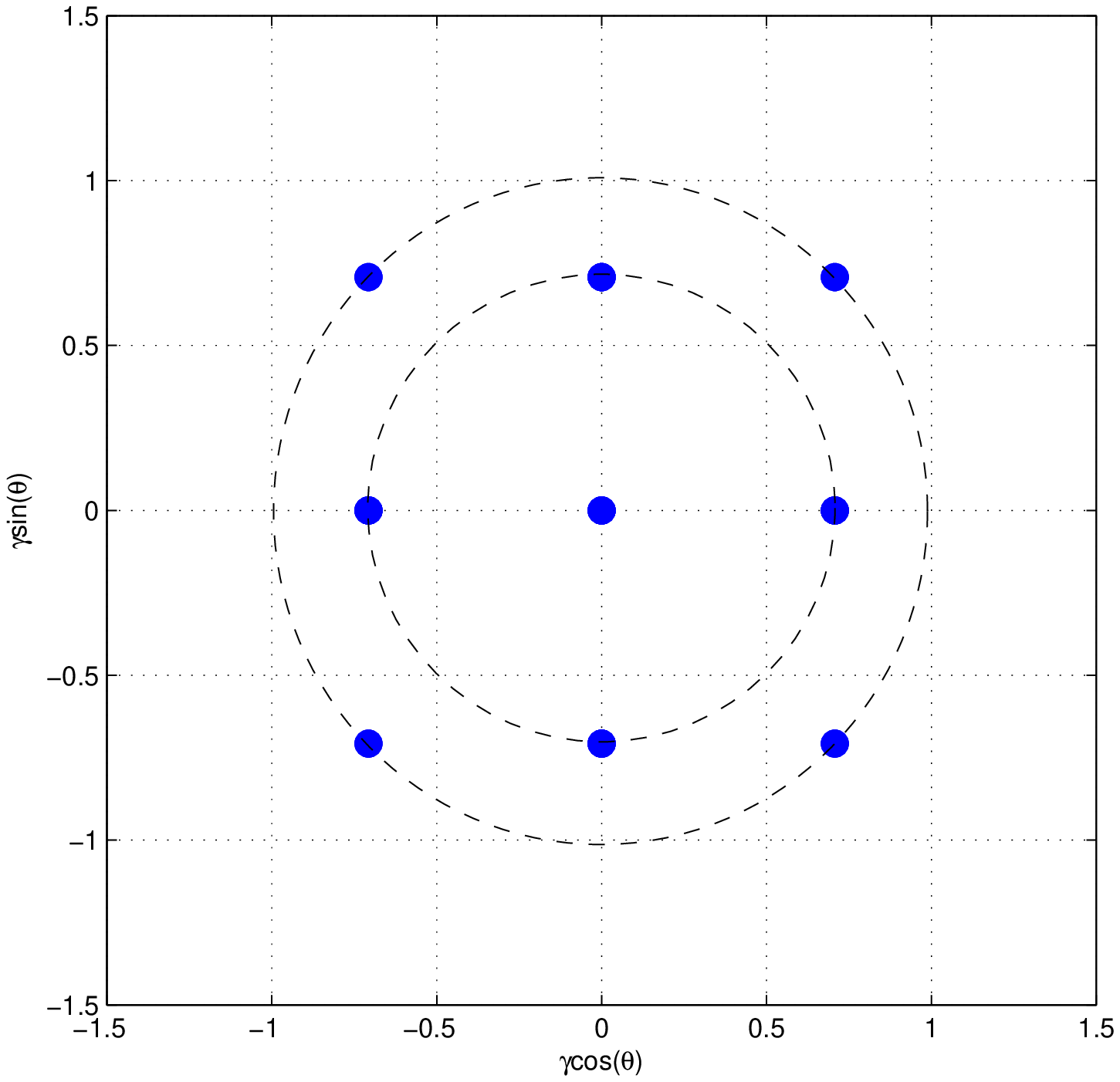}
  \caption{QPSK-BPSK}
  \label{sfs:sub1}
\end{subfigure}%
\begin{subfigure}{.33\textwidth}
  \centering
  \includegraphics[width=15pc]{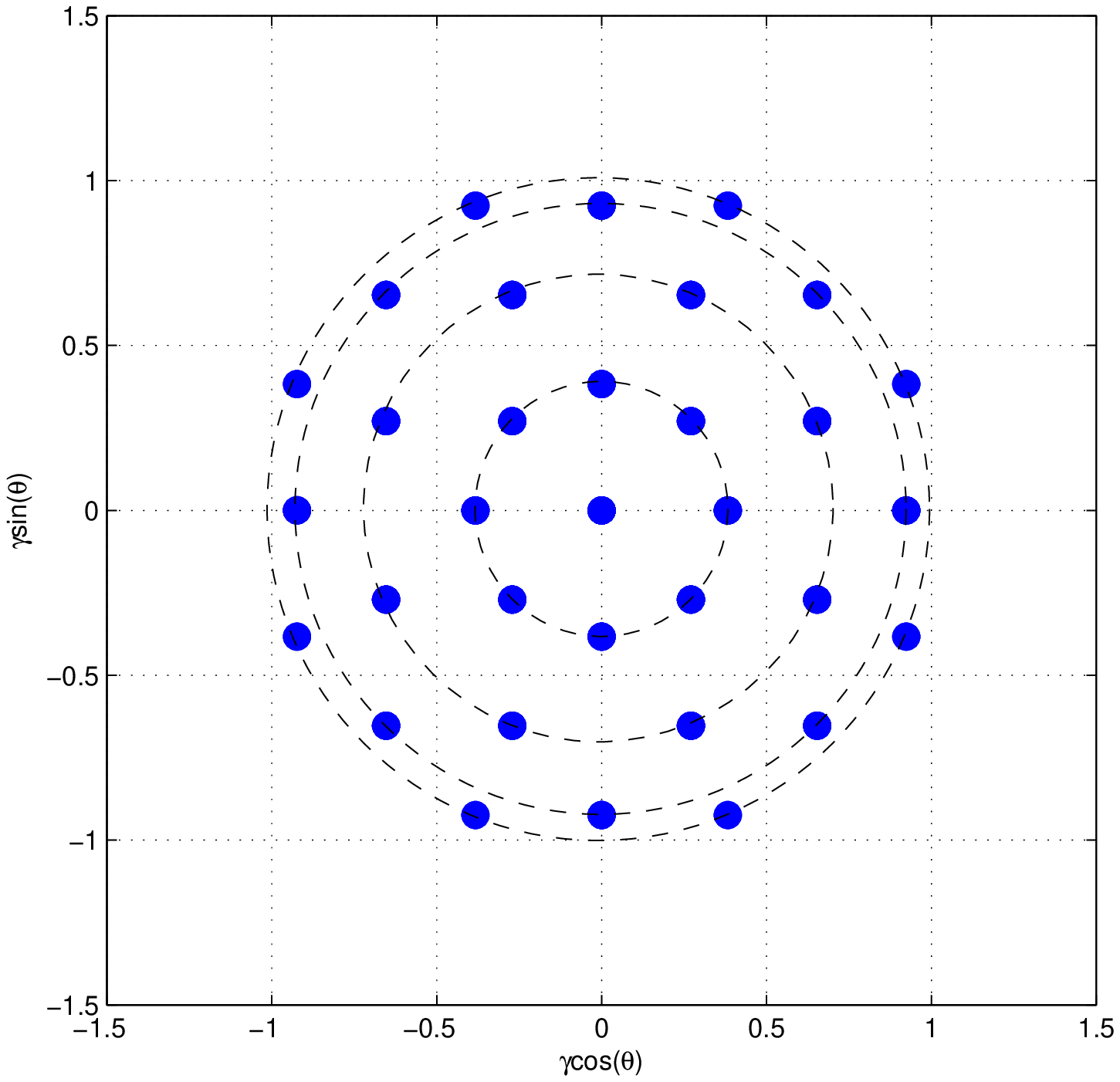}
  \caption{8PSK-BPSK}
  \label{sfs:sub2}
\end{subfigure}
\begin{subfigure}{.33\textwidth}
  \centering
  \includegraphics[width=15pc]{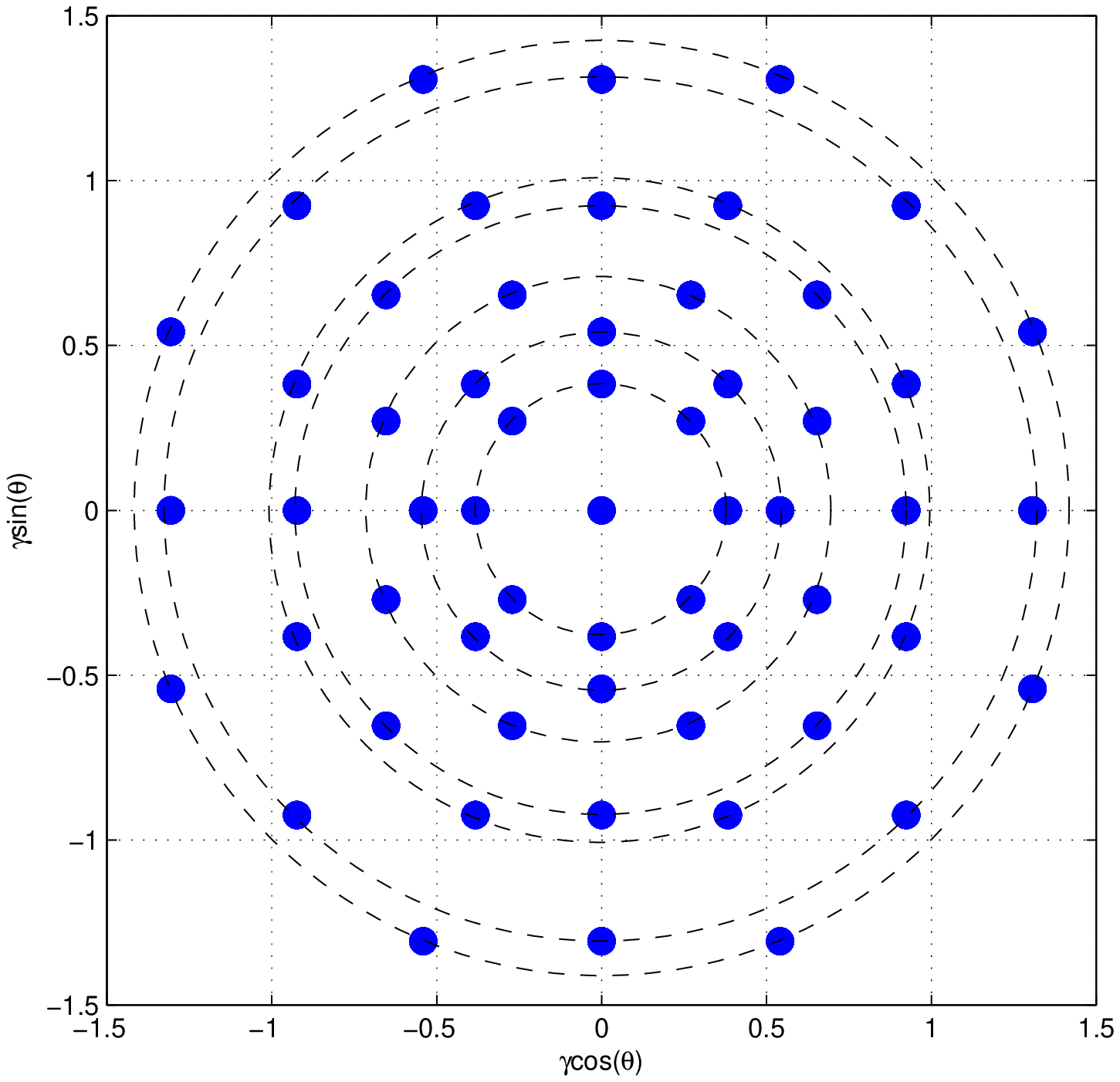}
  \caption{8PSK-QPSK}
  \label{sfs:sub3}
\end{subfigure}
\centering
\caption{Singular Fade States for Different Heterogeneous PSK Modulations}
\label{sfs:test}
\end{figure*}

\section{Construction of Latin Rectangles}
\begin{definition}
A Latin Rectangle of order $M \times N$ on the symbols from the set $\mathbb{Z}_{q} = \left\lbrace 0,1,...,q-1 \right\rbrace$ is an $M \times N$ array, where each cell contains one symbol and each symbol occurs at most once in each row and column.
\end{definition}

From Example 1, there are nine SFS in QPSK-BPSK PLNC. The effect of SFS at $\gamma = 0$ cannot be eliminated. This is because, no matter what constraint is imposed on the mapping, the small value of $\gamma$ leads to very short minimum Euclidean distance. To eliminate the rest, we list the Singularity Removal Constraints for each SFS for QPSK-BPSK PLNC in Table \ref{table_1}. Based on the constraints and requirements of the Exclusive law, we find the smallest possible set of Latin rectangles that remove all the non-zero SFS. A list of Latin rectangles is provided in Table \ref{qpsk_bpsk_ls}. A set of 8 mappings has also been obtained to remove the effect of 32 SFS in 8PSK-BPSK PLNC, which are shown in Table \ref{8psk_bpsk_ls}.
\begin{table}[!htbp]
\caption{Singularity Removal Constraints for QPSK-BPSK PLNC}
\label{table_1}
\centering
\begin{threeparttable}
\begin{tabular}{c c c }
\hline 
\multicolumn{2}{c}{Singular Fade State} & Constraint$^*$ \\
$\gamma_s$ & $\theta_s$ & \\ 
\hline 
0 & 0 & - \\ 
$\frac{1}{\sqrt{2}}$ & 0 & $\mathcal{M}(0,1) = \mathcal{M}(1,0)$; $\mathcal{M}(3,1) = \mathcal{M}(2,0)$ \\ 
$\frac{1}{\sqrt{2}}$ & $\frac{\pi}{2}$ & $\mathcal{M}(0,0) = \mathcal{M}(1,1)$; $\mathcal{M}(2,1) = \mathcal{M}(3,0)$  \\ 
$\frac{1}{\sqrt{2}}$ & $\pi$ & $\mathcal{M}(1,1) = \mathcal{M}(2,0)$; $\mathcal{M}(0,1) = \mathcal{M}(3,0)$  \\ 
$\frac{1}{\sqrt{2}}$ & $\frac{3\pi}{2}$ & $\mathcal{M}(0,0) = \mathcal{M}(3,1)$; $\mathcal{M}(1,0) = \mathcal{M}(2,1)$  \\ 
1 & $\frac{\pi}{4}$ & $\mathcal{M}(0,1) = \mathcal{M}(2,0)$ \\ 
1 & $\frac{3\pi}{4}$ & $\mathcal{M}(1,1) = \mathcal{M}(3,0)$ \\ 
1 & $\frac{5\pi}{4}$ & $\mathcal{M}(2,1) = \mathcal{M}(0,0)$ \\ 
1 & $\frac{7\pi}{4}$ & $\mathcal{M}(3,1) = \mathcal{M}(1,0)$ \\ 
\hline 
\end{tabular}
\begin{tablenotes}[para,flushleft]
  $^*\mathcal{M}^{\gamma,\theta}$ is denoted as $\mathcal{M}$ for compactness. 
\end{tablenotes}
\end{threeparttable}
\end{table} 
\begin{table}[!htbp]
\caption{Mappings for QPSK-BPSK PLNC}
\label{qpsk_bpsk_ls}
\centering
\begin{threeparttable}
\begin{tabular}{ccccccc}
  \hline 
    Map$^*$ &  & 0 & 1 & 2 & 3  & Removes SFS ($\gamma, \theta$)\\ 
  \hline 
   & 0 & 0 & 1 & 2 & 3 & • \\ 
 $\mathcal{C}_1$ & 1 & 1 & 0 & 3 & 2 & $(\frac{1}{\sqrt{2}},0)$ and $(\frac{1}{\sqrt{2}},\pi)$ \\ 
 $\mathcal{C}_2$ & 1 & 3 & 2 & 1 & 0 & $(\frac{1}{\sqrt{2}},\frac{\pi}{2})$ and $(\frac{1}{\sqrt{2}},\frac{3\pi}{2})$ \\ 
$\mathcal{C}_3$ & 1 & 2 & 3 & 0 & 1 & $(1,\frac{\pi}{4})$, $(1,\frac{3\pi}{4})$, $(1,\frac{5\pi}{4})$ and $(1, \frac{7\pi}{4})$ \\ 
  \hline   
   \end{tabular}
\begin{tablenotes}[para,flushleft]
  $^*$The first row is common for all maps. 
\end{tablenotes}
\end{threeparttable}
\end{table}  
\begin{table}[!htbp]
\renewcommand{\arraystretch}{1.3}
\caption{Mappings for 8PSK-BPSK PLNC}
\label{8psk_bpsk_ls}
\centering
\begin{threeparttable}
\begin{tabular}{cccccccccc}
\hline 
Map$^*$ & & 0 & 1 &  2 & 3 & 4 & 5 & 6 & 7 \\ 
\hline 
& 0 & 0 & 1 & 2 & 3 & 4 & 5 & 6 & 7 \\ 

$\mathcal{C}_1$ & 1 & 1 & 5 & 6 & 7 & 3 & 4 & 2 & 0 \\ 

$\mathcal{C}_2$ & 1 & 3 & 0 & 1 & 2 & 5 & 6 & 7 & 4 \\ 

$\mathcal{C}_3$ & 1 & 7 & 2 & 3 & 4 & 1 & 0 & 5 & 6 \\ 

$\mathcal{C}_4$ & 1 & 2 & 7 & 0 & 5 & 6 & 3 & 4 & 1 \\ 

$\mathcal{C}_5$ & 1 & 6 & 3 & 4 & 1 & 2 & 7 & 0 & 5 \\ 

$\mathcal{C}_6$ & 1 & 5 & 4 & 7 & 6 & 1 & 0 & 3 & 2 \\ 

$\mathcal{C}_7$ & 1 & 3 & 6 & 5 & 0 & 7 & 2 & 1 & 4 \\ 

$\mathcal{C}_8$ & 1 & 4 & 5 & 6 & 7 & 0 & 1 & 2 & 3 \\ 
\hline
\end{tabular} 
\begin{tablenotes}[para,flushleft]
  $^*$The first row is common for all maps. 
\end{tablenotes}
\end{threeparttable}
\end{table}
 
\begin{figure*}[!htbp]
\centering
\begin{subfigure}{.5\textwidth}
  \centering
  \includegraphics[width=21pc]{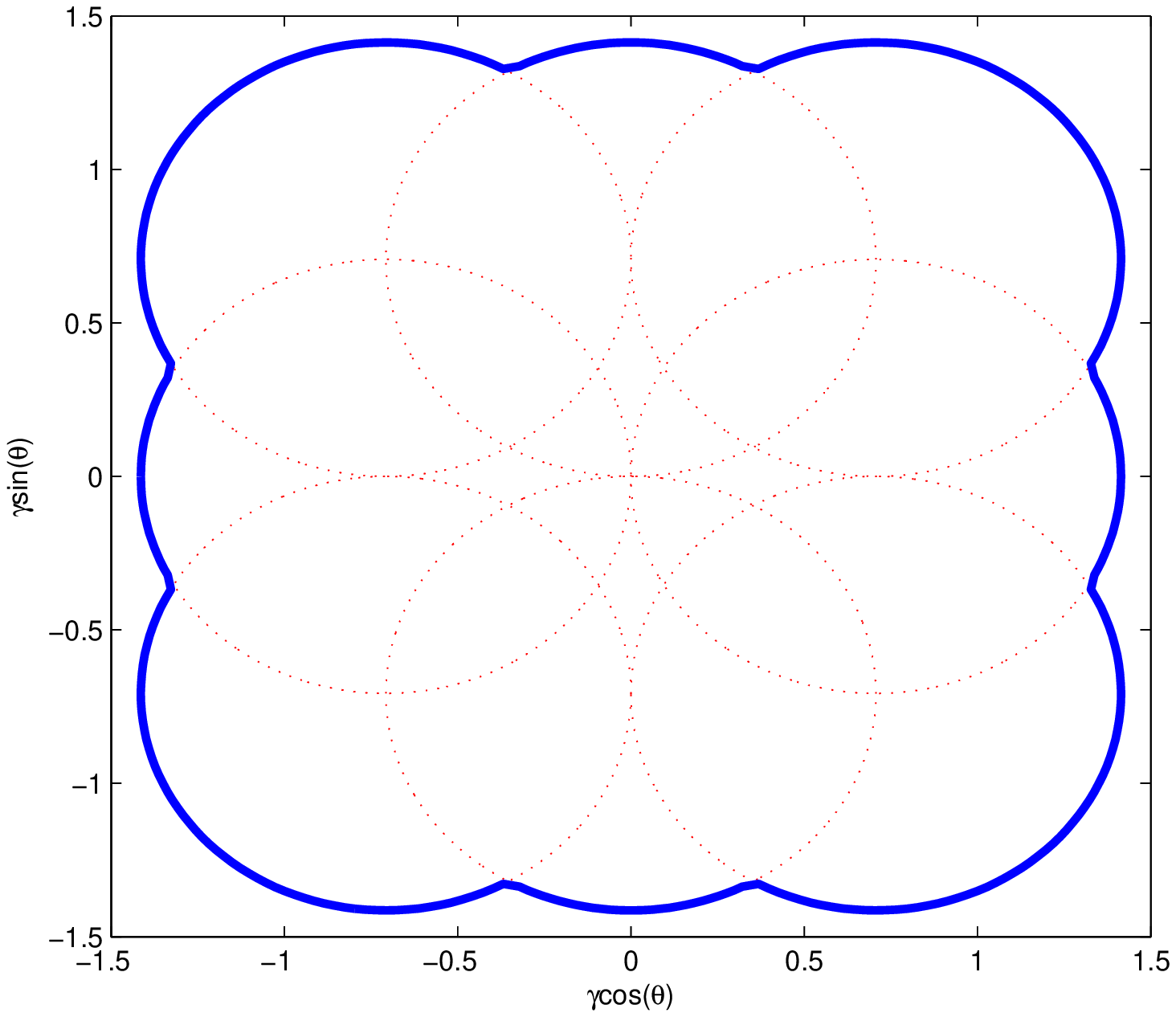}
  \caption{QPSK-BPSK}
  \label{cir:sub1}
\end{subfigure}%
\begin{subfigure}{.5\textwidth}
  \centering
  \includegraphics[width=21pc]{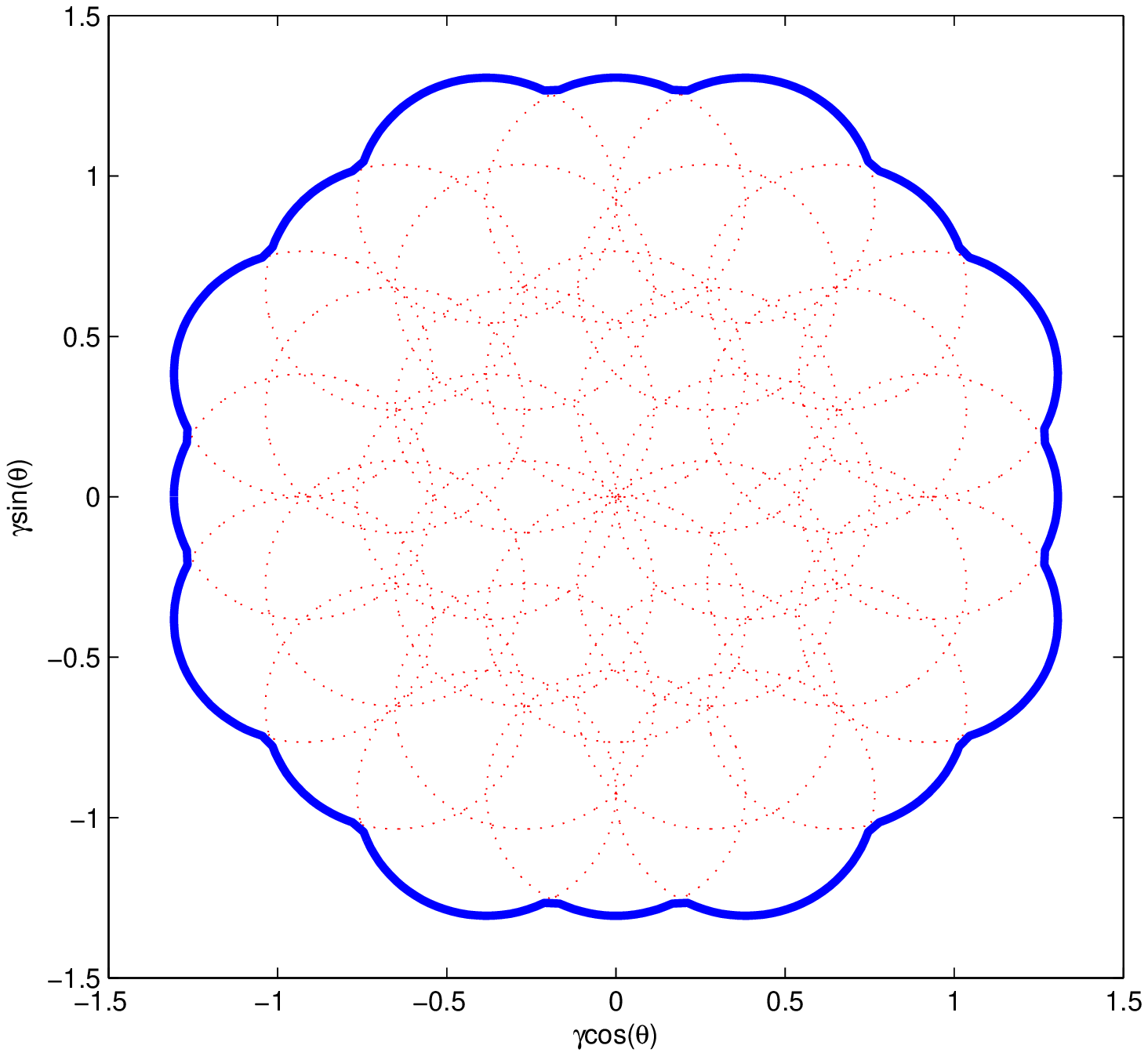}
  \caption{8PSK-BPSK}
  \label{cir:sub2}
\end{subfigure}
\centering
\caption{External Clustering Independent Region (Outside blue line).}
\label{cir:test}
\end{figure*}
\begin{figure*}[!htbp]
\centering
\begin{subfigure}{.5\textwidth}
  \centering
  \includegraphics[width=21pc]{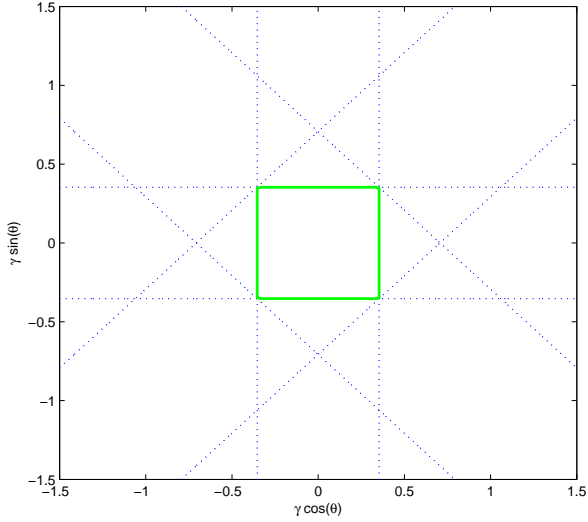}
  \caption{QPSK-BPSK}
  \label{cir_int:sub1}
\end{subfigure}%
\begin{subfigure}{.5\textwidth}
  \centering
  \includegraphics[width=21pc]{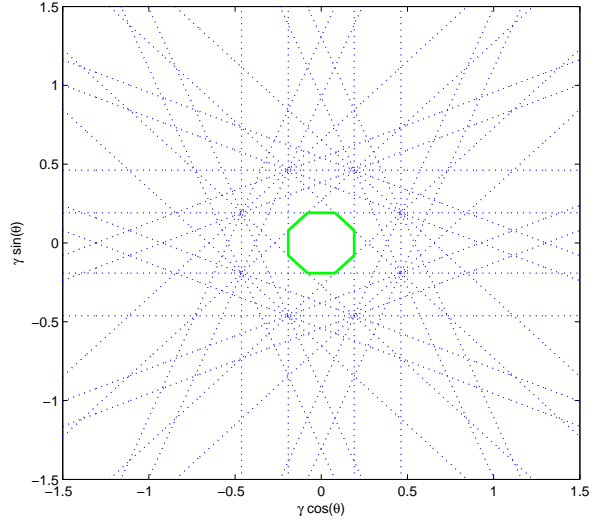}
  \caption{8PSK-BPSK}
  \label{cir_int:sub2}
\end{subfigure}
\centering
\caption{Internal Clustering Independent Region (Inside green line).}
\label{cir_int:test}
\end{figure*}
\subsection{Clustering Dependent Region}
 \begin{figure}[!htbp]
 \centering
 \includegraphics[width=21pc]{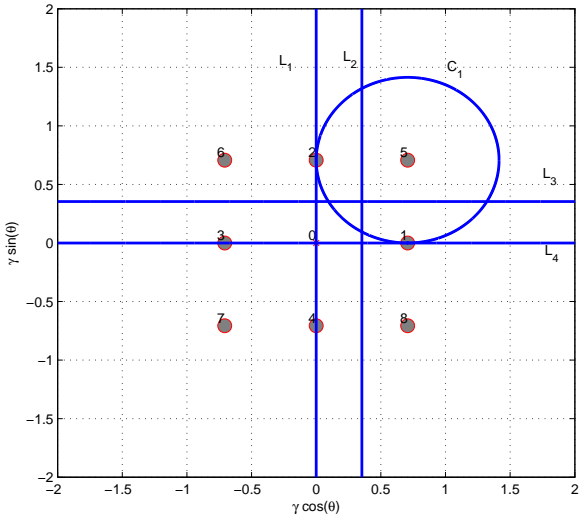}
 \caption{Obtaining clustering dependent region for the $\frac{1}{\sqrt{2}}(1+j)$ SFS in QPSK-BPSK PLNC.}
 \label{reg_example}
 \end{figure}
 \begin{figure*}[!htbp]
  \centering
  \begin{subfigure}{.5\textwidth}
    \centering
    \includegraphics[width=21pc]{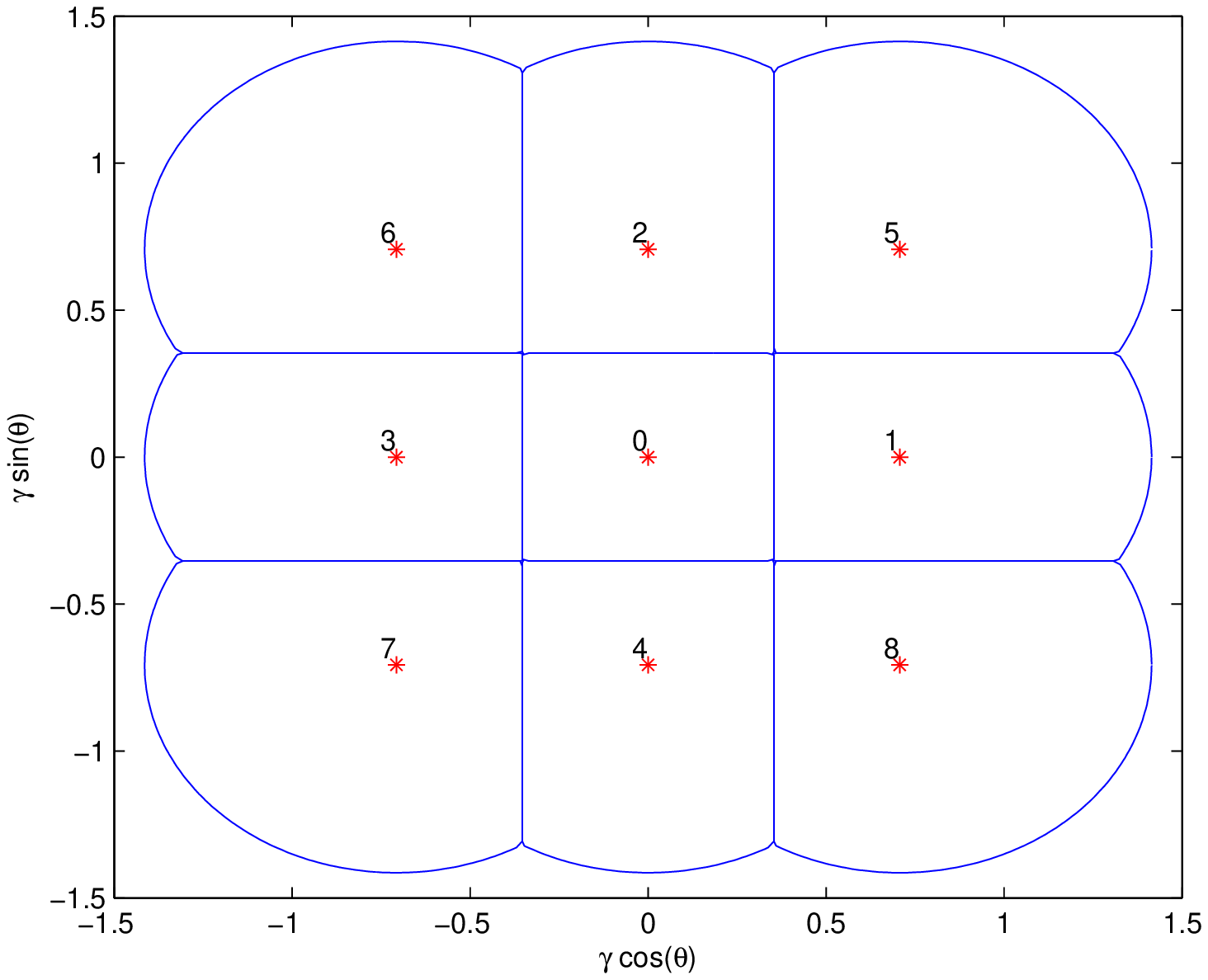}
    \caption{QPSK-BPSK}
    \label{reg:sub1}
  \end{subfigure}%
  \begin{subfigure}{.5\textwidth}
    \centering
    \includegraphics[width=21pc]{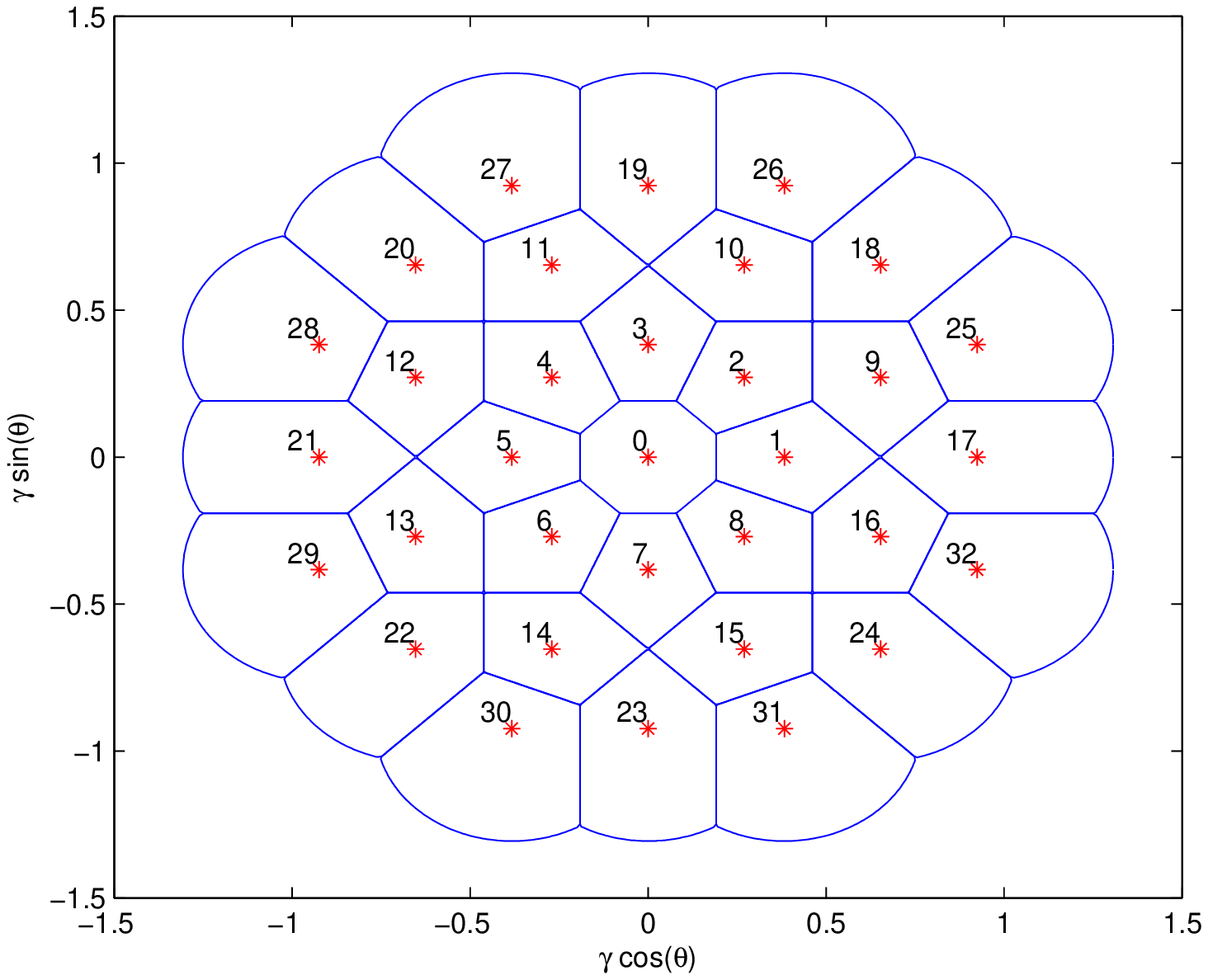}
    \caption{8PSK-BPSK}
    \label{reg:sub2}
  \end{subfigure}
  \centering
  \caption{Clustering Dependent Regions for PLNC with heterogeneous PSK modulations.}
  \label{reg:test}
  \end{figure*}
\section{Quantization of Fade State Plane}
The mappings in the previous section removed the particular fade state. However, the best possible mapping for each $(\gamma,\theta)$ has to be found. We follow the scheme proposed by \cite{muralidharan2013wireless} for clustering at the relay, which is based on using the maps which were used for the removal of SFS. For a given realization of $\gamma e^{j\theta}$ we use one of the Latin Rectangles which is used to remove an SFS based on the criteria given below \cite{muralidharan2013wireless}.  
\begin{definition} The distance metric $\mathcal{D}$ is defined as
\begin{align}
\mathcal{D}(\gamma,\theta,d_1,d_2) \triangleq |d_1+&\gamma e^{j\theta}d_2 | , \nonumber\\ &(0,0) \neq (d_1 ,d_2) \in \mathcal{S}_1 \times \mathcal{S}_2.
\end{align}
\end{definition}
Let $(d_1(h),d_2(h))=(d_1 ,d_2)\in\mathcal{S}_1\times \mathcal{S}_2$, where $\frac{-d_1}{d_2}=h$. 
The criterion for selecting the SFS, whose map is to be used for a given fade state is as follows: \\ 
\emph{If $\underset{d_1,d_2}{argmin}~\mathcal{D}(\gamma,\theta,d_1,d_2) = (d_1(h),d_2(h))$ then choose the clustering which removes the singular fade state h.}
\subsection{Clustering Independent Region}
\begin{definition}
\text{(Muralidharan et al. \cite{muralidharan2013wireless})} The set of all values of fade states for which any clustering satisfying the exclusive law gives the same minimum cluster distance is called Clustering Independent Region.
\end{definition} First, an upper bound on the minimum cluster distance as in \cite{muralidharan2013wireless} is obtained. The existence of such a region is proved by adapting Lemma 10 from \cite{namboodiri2013physical} for the HePNC scenario.
\begin{theorem}
For any clustering $\mathcal{C}$ satisfying the exclusive law, with unit energy $M_i$-PSK signal sets, $i=1, 2$, $d_{min}(\mathcal{C};\gamma,\theta)$ is upper-bounded as,
\begin{equation}
d_{min}(\mathcal{C};\gamma,\theta) \leq min \left\lbrace 2sin\left(\frac{\pi}{M_1}\right), 2\gamma sin\left(\frac{\pi}{M_2}\right) \right\rbrace.
\end{equation}
\end{theorem}
\begin{proof}
As $\mathcal{C}$ satisfies the exclusive law, $\mathcal{L_\mathcal{C}}(x_A,x_B) \neq \mathcal{L_\mathcal{C}}(x_A,x_B')$ where $x_A \in \mathcal{S}_1$ and  $x_B , x_B' \in \mathcal{S}_2$. We have 
\begin{align}
& d_{min}(\mathcal{C};\gamma,\theta) = \nonumber \\ & \underset{ \underset{(x_A,x_B) \neq (x_A',x_B')\in (\mathcal{S}_1 \times \mathcal{S}_2)}{\mathcal{L}_{\mathcal{C}}(x_A,x_B) \neq \mathcal{L}_{\mathcal{C}}(x_A',x_B')}}{min} |(x_A-x_A')+ \gamma e^{j\theta} (x_B-x_B')| \nonumber \\
  & \leq \underset{(x_A,x_B) \neq (x_A,x_B') \in \mathcal{S}_1 \times \mathcal{S}_2 }{min} \gamma |e^{j\theta} (x_B-x_B')| \nonumber \\
    & = \gamma \underset{x_B \neq x_B' \in \mathcal{S}_2}{min}|(x_B-x_B')| = 2\gamma sin(\frac{\pi}{M_2}).
\end{align}
Similarly, from the fact that $\mathcal{L_{\mathcal{C}}}(x_A,x_B) \neq \mathcal{L_{\mathcal{C}}}(x_A',x_B)$, where $x_A, x_A' \in \mathcal{S}_1 ,x_B \in \mathcal{S}_2$ and $x_A \neq x_A'$, we have  $d_{min}(\mathcal{C};\gamma,\theta) \leq 2\sin(\frac{\pi}{M_1})$. 
\end{proof}
From the definitions, $d_{min}(\gamma e^{j\theta}) \leq d_{min}(\mathcal{C};\gamma,\theta)$. Hence, from Theorem 1, it follows that regardless of which $\mathcal{C}$ is considered, $d_{min}(\mathcal{C};\gamma,\theta) = 2sin(\pi/M_1)$ when $\gamma \gg 1$. Similarly, for $\gamma \ll 1$, $d_{min}(\mathcal{C};\gamma,\theta) = 2\gamma sin(\pi/M_2)$.     

Let $\mathcal{R}^{ext}_{CI}$ and $\mathcal{R}^{int}_{CI}$ denote the clustering independent regions corresponding to $\gamma>1$ and $\gamma<1$ respectively. From Theorem 1, for $\gamma>1$ , $\min(2sin(\pi/M_1),2\gamma sin(\pi/M_2)) = 2sin(\pi/M_1)$. Hence,
\begin{align}
\mathcal{R}^{ext}_{CI} =  &\lbrace \gamma e^{j\theta} : |x_{k_1,n_1,1}+\gamma e^{j\theta} x_{k_2,n_2,2}| \geq 2sin(\pi/M_1), \nonumber\\
&\forall (0,0) \neq (x_{k_1,n_1,1},x_{k_2,n_2,2}) \in \Delta\mathcal{S}_1 \times \Delta\mathcal{S}_2, \nonumber\\ &\gamma > 1 , -\pi \leq \theta < \pi \rbrace. 
\label{ext_ci_eq2}
\end{align}

Let $c_{l_1,l_2}$, $1 \leq 2^{\delta}l_2 \leq l_1 \leq  M_1/2$ and $1 \leq l_2 \leq M_2/2$ denote the circle centered at the origin with radii $r_{l_1,l_2} = sin(l_1 \pi /M_1)/sin(l_2 \pi /M_2)$. Let $C_{l_1,l_2}$ denote the set of circles whose centers are the SFS which lie on $c_{l_1,l_2}$ and have radii $r_{l_2} = sin(\pi /M_1)/sin(2^{\delta} l_2 \pi/M_1)$. The following theorem generalizes the boundary of external $CI$ region given in \cite{muralidharan2013wireless}.

\begin{theorem}
The region $\mathcal{R}^{ext}_{CI}$ is the common outer envelope  region formed by the circles $C_{l_1,l_2}$, $1 \leq 2^{\delta}l_2 \leq l_1 \leq  M_1/2$.
\end{theorem}
\begin{proof}
See Appendix.
\end{proof}
Similarly, 
\begin{align}
\mathcal{R}^{int}_{CI} =  &\lbrace \gamma e^{j\theta} : |x_{k_1,n_1,1}+\gamma e^{j\theta} x_{k_2,n_2,2}| \geq 2 \gamma sin(\pi/M_2), \nonumber\\
&\forall (0,0) \neq (x_{k_1,n_1,1},x_{k_2,n_2,2}) \in \Delta\mathcal{S}_1 \times \Delta\mathcal{S}_2, \nonumber\\ &\gamma < 1 , -\pi \leq \theta < \pi \rbrace
\label{int_ci_eq}
 \end{align}

The transformation $\gamma' e^{j\theta'} = \frac{1}{\gamma e^{j\theta}}$, is called complex inversion. It can be verified that by applying complex inversion in (\ref{ext_ci_eq2}) we do not get (\ref{int_ci_eq}) unless $M_1 = M_2$. Unlike Lemma 11 of \cite{muralidharan2013wireless}, getting internal clustering independent region for heterogeneous PLNC with PSK is non-trivial. The following theorem gives the general method to obtain the internal independent region for PLNC using heterogeneous PSK modulations.

\begin{theorem}
The region $\mathcal{R}^{int}_{CI}$ is the region formed by the intersection of following regions \\
\begin{equation}
|\gamma e^{j\theta} - c_{Int}| \geq r_{Int}, ax+by \leq c'_{Int},
\label{equation_ineq1}
\end{equation}
where,  
\begin{align}
c_{Int} &= -\left(\frac{x_{k1,n1,1}}{x_{k2,n2,2}}\right)\left(\frac{sin(n_2 \pi/M_2)^2}{sin(n_2 \pi/M_2)^2-sin(\pi/M_2)^2}\right) \\
r_{Int} &= \frac{sin(n_1 \pi/M_1)sin(\pi/M_2)}{2(sin(n_2 \pi/M_2)^2-sin(\pi/M_2)^2)}
\end{align}
$\forall (k_i,n_i)$ such that $0 \leq k_i\leq M_i-1$,$1 \leq n_i \leq M_i/2$ and  
 \begin{align}
 a_{Int} &= \mathit{Re}(x^{*}_{k_1,n_1,1}x_{k_1,1,2}) , \\
 b_{Int} &= \mathit{Im}(x^{*}_{k_1,n_1,1}x_{k_1,1,2}) , \\
 c'_{Int} &= |x_{k_1,n_1,1}|^2 /2 , 
 \end{align}
 where $x = \mathit{Re}(\gamma e^{j\theta})$, $y = \mathit{Im}(\gamma e^{j\theta})$ and * denotes complex conjugation.
\end{theorem}
\begin{proof}
From (\ref{int_ci_eq}), $\mathcal{R}^{int}_{CI}$, is the intersection of all regions in the complex fade state plane which satisfy the inequality 
\begin{equation}
|x_{k_1,n_1,1}+\gamma e^{j\theta} x_{k_2,n_2,2}| \geq 2 \gamma sin(\pi/M_2).
\end{equation}
 By squaring on both sides of the above inequality and completing the magnitude square, we get the curves given in (\ref{equation_ineq1}).  
\end{proof}

 The region other than the clustering independent region is called clustering dependent region. From the criterion to associate a clustering to any given channel fade state, it is seen that every SFS has an associated region in the channel fade state plane in which the clustering which removes that SFS is used at the relay. Hence the region associated with the SFS $h$ is, 
 \begin{align}
 \mathcal{R}_{h} = 
\lbrace \gamma e^{j\theta} : \underset{(d_1,d_2) \in \Delta\mathcal{S}_1 \times \Delta\mathcal{S}_2}{argmin} &\mathcal{D}(\gamma,\theta,d_1,d_2) = (d_1',d_2') \nonumber\\ & where -d_1'/d_2' = h \rbrace.
 \end{align}
 From Lemma 2, we know that the $M_1$ SFS which lie on the same circle, have phase angles of the form $\frac{2 \pi l}{M_1}$ or $\frac{\pi}{M_1} + \frac{2 \pi l}{M_1}$ for $0 \leq l \leq M_1$. Hence, there is an angular symmetry of $\frac{2\pi}{M_1}$. Therefore, it suffices to consider those SFS which lie on the lines $\theta = 0$ and $\theta = \pi/M_1$ and use symmetry to obtain the regions $\mathcal{R}_{h}$ for all values of $h$.

The pairwise transition boundary formed by the SFS $h$ and $h'$, denoted by $c(h,h')$, is the set of values of $\gamma e^{j\theta}$ for which $|d_1(h)+\gamma e^{j\theta}d_2(h)|=|d_1(h')+\gamma e^{j\theta}d_2(h')|$.
\begin{theorem}
The pairwise transition boundaries are either circles or straight lines.
\end{theorem}
\begin{proof}
We reproduce the proof provided by Namboodiri et al. [Theorem 2, \cite{namboodiri2013physical}] for completeness. The curve $c(h_1,h_2)$ is,
\begin{equation}
|d_1(h_1) + \gamma e^{j\theta} d_2(h_1)| = |d_1(h_2) + \gamma e^{j\theta} d_2(h_2)|.
\end{equation}
Squaring on both sides and manipulating, we get the following:
\begin{align}
|d_1(h_1) + \gamma e^{j\theta} d_2(h_1)|^2 &= |d_1(h_2) + \gamma e^{j\theta} d_2(h_2)|^2. \\
\gamma^2 (|d_2(h_1)|^2 - |d_2(h_2)|^2) &- 2\Re \{ \gamma e^{j\theta} (d_1^{*}(h_2)d_2(h_2) \nonumber \\
-d_1^{*}(h_1)d_2(h_1)) \} &= |d_1(h_2)|^2 - |d_1(h_1)|^2. 
\end{align}
Assume $|d_2(h_1)|^2 \neq |d_2(h_2)|^2$. Dividing by $|d_1(h_2)|^2 - |d_1(h_1)|^2 $, completing the magnitude square on the LHS and substituting $h_i = \frac{-d_1(h_i)}{d_2(h_i)}$, we get,
\begin{equation}
|\gamma e^{j\theta} - c| = r, 
\end{equation}
where
\begin{align}
c &= \frac{-h_2|d_2(h_2)|^2 + h_1|d_2(h_1)|^2}{|d_2(h_1)|^2 - |d_2(h_2)|^2},\\
r &= \frac{|d_2(h_1)d_2(h_2)(h_1 - h_2)|}{|d_2(h_1)|^2 - |d_2(h_2)|^2}.
\end{align}
If $|d_2(h_1)|^2 = |d_2(h_2)|^2$, we get the equation of a straight line $ax + by = c'$, where $x = \mathit{Re}(\gamma e^{j\theta})$, $y = \mathit{Im}(\gamma e^{j\theta})$ and 
\begin{align}
a &= \mathit{Re}(h_2)|d_2(h_2)|^2 - \mathit{Re}(h_1)|d_2(h_1)|^2, \\
b &= \mathit{Im}(h_2)|d_2(h_2)|^2 - \mathit{Im}(h_1)|d_2(h_1)|^2, \\
c' &= |d_1(h_2)|^2 - |d_1(h_1)|^2. \
\end{align}
\begin{figure*}[!htbp]
\centering
\begin{subfigure}{.5\textwidth}
  \centering
  \includegraphics[width=21pc]{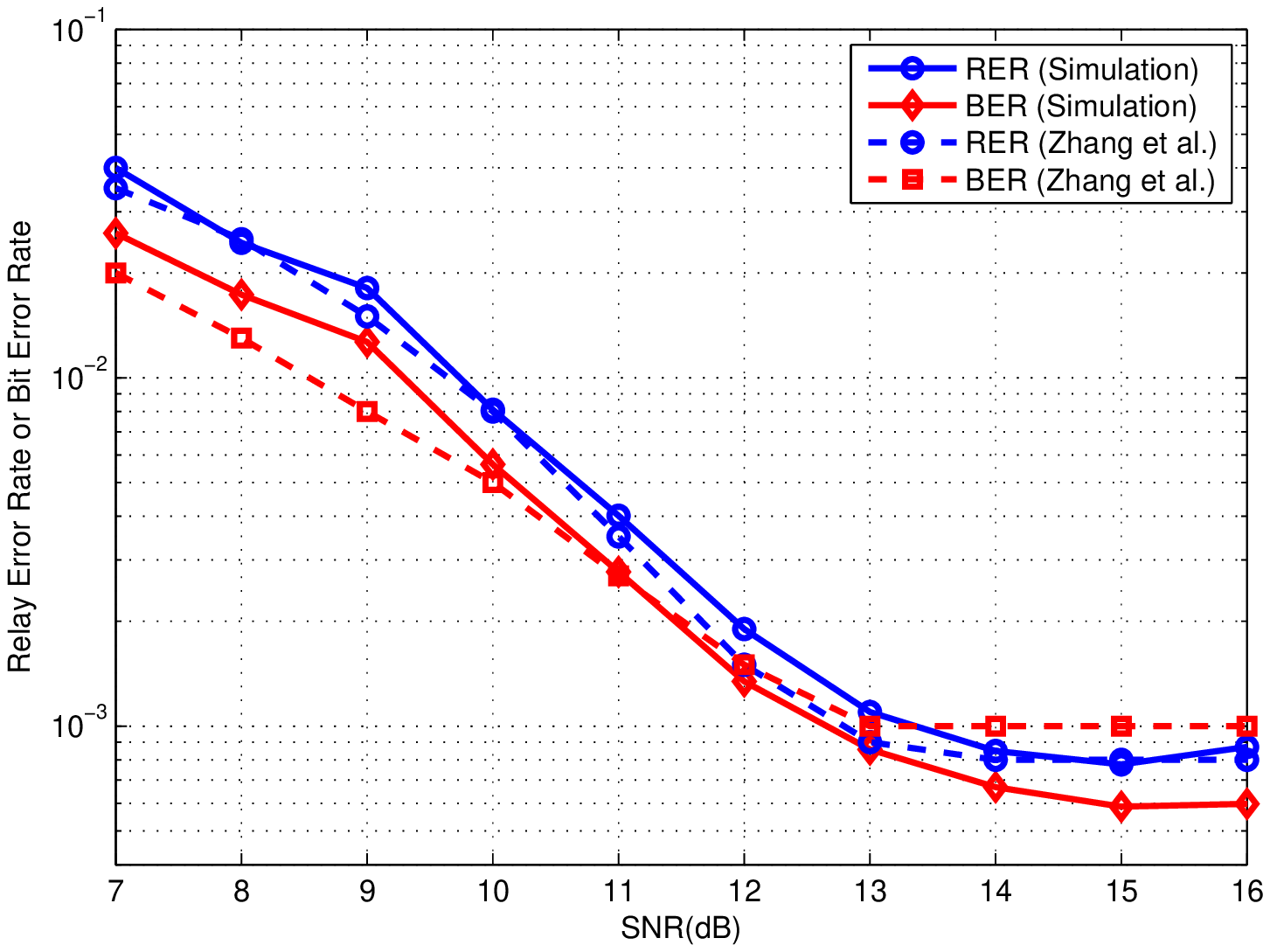}
  \caption{AWGN Channel}
  \label{ber_qpsk:sub1}
\end{subfigure}%
\begin{subfigure}{.5\textwidth}
  \centering
  \includegraphics[width=21pc]{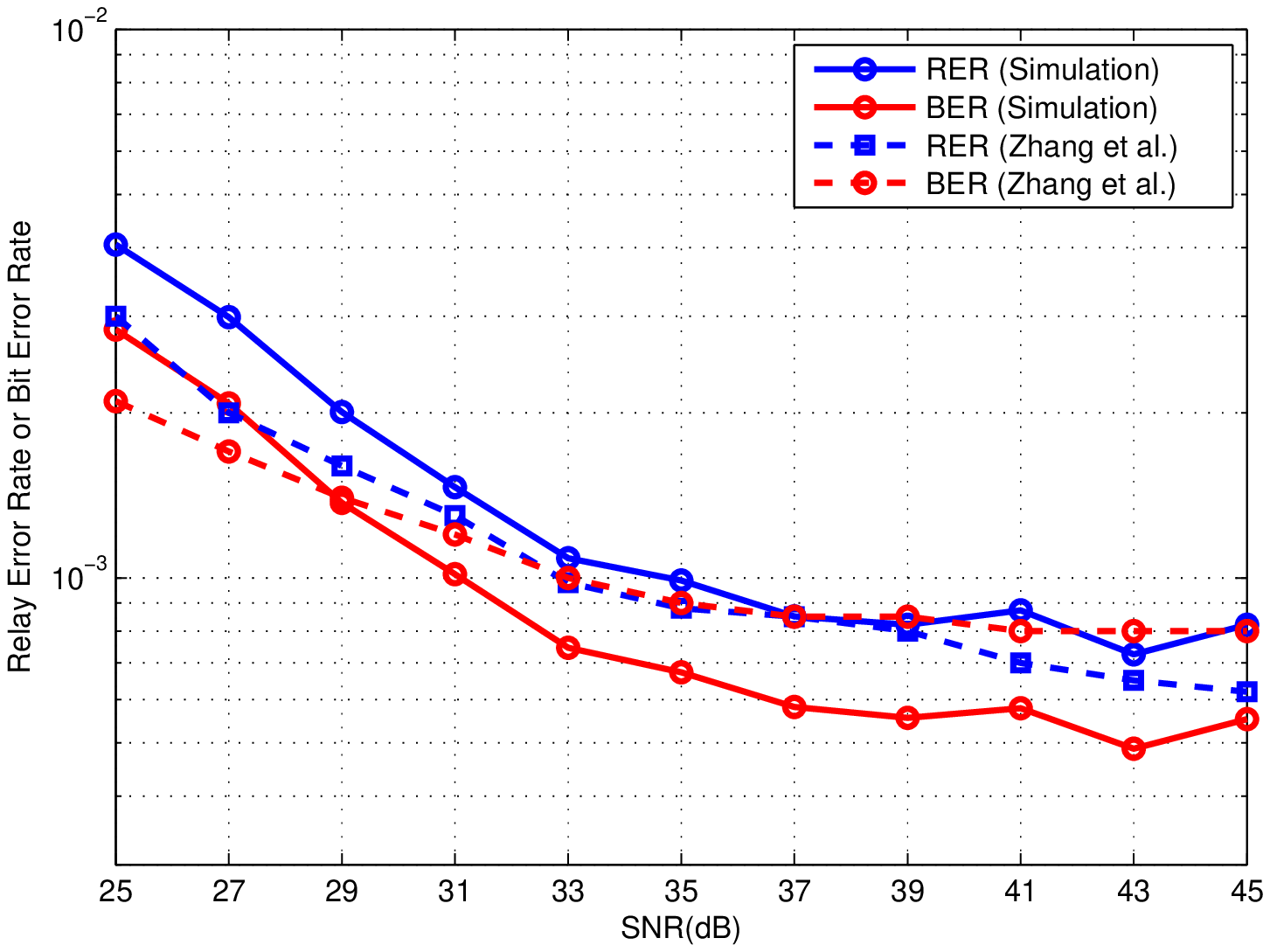}
  \caption{Rayleigh Fading Channel}
  \label{ber_qpsk:sub2}
\end{subfigure}
\centering
\caption{BER and RER for QPSK-BPSK PLNC}
\label{ber_qpsk:test}
\end{figure*}
\begin{figure*}[!htbp]
\centering
\begin{subfigure}{.5\textwidth}
  \centering
  \includegraphics[width=21pc]{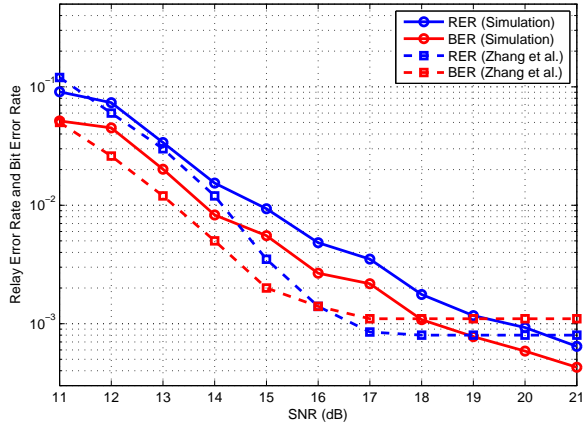}
  \caption{AWGN Channel}
  \label{ber_8psk:sub1}
\end{subfigure}%
\begin{subfigure}{.5\textwidth}
  \centering
  \includegraphics[width=21pc]{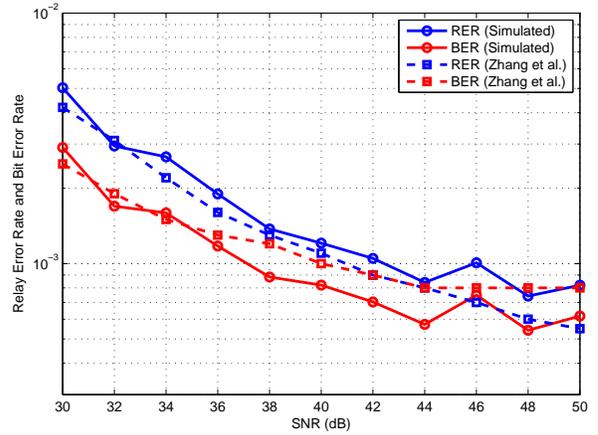}
  \caption{Rayleigh Fading Channel}
  \label{ber_8psk:sub2}
\end{subfigure}
\centering
\caption{BER and RER for 8PSK-BPSK PLNC}
\label{ber_8psk:test}
\end{figure*}
This proof is independent of the signal sets considered at the users.
\end{proof}
It can be easily verified that the following lemmas from \cite{muralidharan2013wireless}, which are used to get the regions associated with the SFS, can be applied to the case of  heterogeneous PSK modulations also.
\begin{lemma}
The region $\mathcal{R}_h$, where the SFS $h$ lies on the line $\theta = a;a \in \lbrace 0 , \pi/M_1 \rbrace$, lies inside the wedge formed by the lines $\theta = a - \pi/ M_1 $ and  $\theta = a + \pi/ M_1 $.
\end{lemma}
\begin{lemma}
To obtain the boundaries of $\mathcal{R}_h$, where the SFS $h$ lies on the line $\theta = a;a \in \lbrace 0 , \pi/M_1 \rbrace$, it is enough to consider those curves $c(h,h')$, $|h|\neq|h'|$ which lie on the lines $\theta = a - \pi/ M_1$, $\theta = a$  and  $\theta = a + \pi/ M_1 $.
\end{lemma}
\begin{example}
Consider the case where $M_1=4$ and $M_2=2$. From Example 1, there are $9$ SFS. We select one $h \in \mathcal{H}$ and derive the boundaries of clustering dependent region. Let $h = \frac{1}{\sqrt{2}}(1+j)$. The difference constellation points for $h$ are $d_1 = \sqrt{2} + \sqrt{2}j$ and $d_2 = -2$. From Lemma 3, the region lies in the wedge formed by the lines $\theta = 0$ and $\theta = \pi/2$ (since $a=\pi/4$). These lines are denoted by $L_1$ and $L_4$.
 \begin{figure}[!htbp]
   \centering
   \includegraphics[width=21pc]{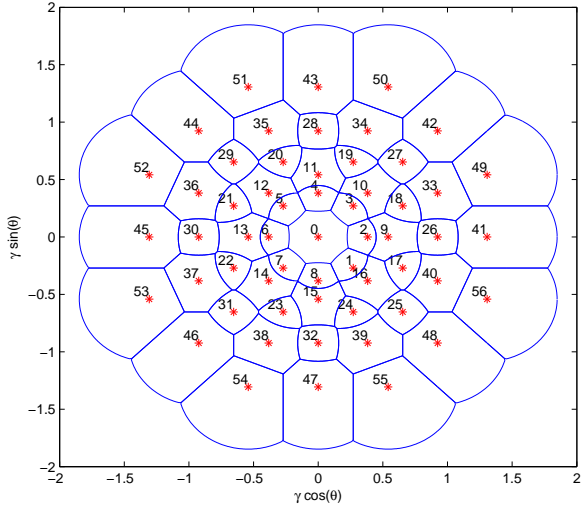}
   \caption{Clustering Dependent Regions for 8PSK-QPSK PLNC}
   \label{8psk_qpsk}
   \end{figure}
  \begin{figure}[!htbp]
 \centering
 \includegraphics[width=21pc]{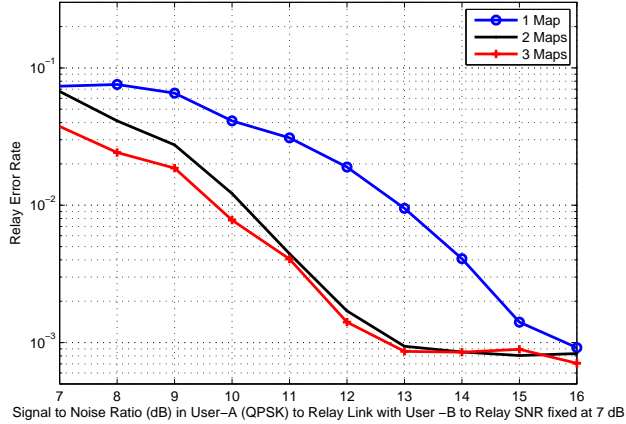}
 \caption{RER of QPSK-BPSK scheme in AWGN channel}
 \label{DIFF_MAP}
 \end{figure}  
We find the equation of the boundary between $h$ and $h' = 1/\sqrt{2}$. The difference constellation points for $h'$ are $d_1' = \sqrt{2}$ and $d_2' = -2$. From Theorem 1, the boundary, $c(h,h')$, satisfies the expression of a straight line (since $d_2 = d_2'$), which is $\gamma sin(\theta) = \frac{1}{2\sqrt{2}}$. This is the horizontal line $L_3$. Similarly, for $h' = j/\sqrt{2}$, $d_1' = \sqrt{2}j$ and $d_2' = -2$, and we get a vertical line $L_2$, which is $\gamma cos(\theta) = \frac{1}{2\sqrt{2}}$. 

Since the region is also bounded by the external clustering independent region, we have to consider the circle $C_1$ centered at $h$ and having radius $1/\sqrt{2}$. All the lines and circles have been shown in Fig. \ref{reg_example}. The desired region around h is the internal  region formed by the intersection of these curves. 
\end{example}

The regions obtained for QPSK-BPSK case are shown in Fig. \ref{reg:sub1} and they match with those in \cite{zhang2016design}. In Fig. \ref{reg:sub2} we show the regions for 8PSK-BPSK scheme. 

\section{Simulation Results}
 The performance of the HePNC scheme is evaluated for AWGN and Rayleigh fading channels. Results are compared with those given in \cite{zhang2016design}. In this paper, AWGN channel indicates a channel with additive Gaussian noise and without phase synchronization of users, i.e., $\gamma = 1$ and $\theta \mytilde \mathit{Unif}[0,2\pi)$. The two performance metrics are Relay Error Rate (RER) and Bit Error Rate (BER). During simulation, the average SNR of the B-R link is kept constant and the average SNR of A-R link is varied. To ensure error probability better than $10^{-3}$ in B-R link, $SNR_{BR} = 7 ~dB$ for AWGN channel and $SNR_{BR} = 25 ~dB$ for Rayleigh fading channel. The performance parameters and the values are chosen to present a comparison with Zhang et al. \cite{zhang2016design}. 
 
 Fig. \ref{DIFF_MAP} shows RER as a function of the number of maps used at relay for QPSK-BPSK PLNC scheme. It can be seen that one denoising map cannot remove all the SFS and hence the performance improves till all three maps are used (i.e all SFS are removed). The error floor behavior is a result of fixing $SNR_{BR}$. In end-to-end BER simulations, if the BER from A to B is $BER_{AB}$ and from B to A is $BER_{BA}$, the overall BER is calculated as the average BER across both users, $BER_{avg} = (2BER_{AB} + BER_{BA})/3$ \cite{zhang2016design}. The RER \& average BER are shown in Fig. \ref{ber_qpsk:test}\ and agree with \cite{zhang2016design}.
 
 For the 8PSK-BPSK PLNC, the average BER is calculated as $BER_{avg} = (3BER_{AB} + BER_{BA})/4$. The overall BER \& RER are shown in Fig. \ref{ber_8psk:test}\ and agree with \cite{zhang2016design}. 
\section{Discussion and Future Work}
The paper extends the framework of \cite{muralidharan2013wireless} to PLNC with heterogeneous PSK modulations. Exact expressions for the number and location of SFS for heterogeneous PSK modulations are given and equations for the boundaries of clustering independent and dependent regions are derived.  

A possible direction for future work would be to investigate QAM-QAM heterogeneous PLNC analytically. Another possible direction is to use the theory of constrained partially filled Latin rectangles for finding the mappings similar to \cite{muralidharan2013wireless}.  
\appendices
\section{Proof of Theorem 2}
\begin{proof}
The proof of Theorem 2 follows by combining the inferences from Lemma 5, 6 and 7 described hereafter.
\end{proof}
Let $p_i = \left\lbrace c_{l,i} | 2^{\delta}i \leq l \leq M_1/2 \right\rbrace$ , $ 1 \leq i \leq M_2/2$, be a group of circles. The region between the outermost circle and the innermost circle in $p_i$ is called the ring formed by $p_i$. Let $c_O$ be the unit circle centered at the origin of fade state plane.
The following lemma is modified version of Lemma 17 in \cite{muralidharan2013wireless}.
\begin{lemma}
The external clustering independent region is the unshaded region obtained when the interior regions of all the circles which belong to the sets $C_{l_1,l_2}$, $1 \leq 2^{\delta}l_2 \leq l_1 \leq  M_1/2$.
\end{lemma} 
\begin{proof}
From (\ref{diff_const_eqn}), $|x_{ki,ni,i}|= 2sin(\pi n_i/M_i)$ for $i = 1,2$.
From (\ref{ext_ci_eq2}) we have
\begin{align}
\mathcal{R}^{ext}_{CI} &=  \lbrace \gamma e^{j\theta} : |\frac{x_{k1,n1,1}}{x_{k2,n2,2}}+\gamma e^{j\theta}| \geq \frac{sin\left(\frac{\pi}{M_1}\right)}{sin\left(\frac{2^{\delta}n_2\pi}{M_1}\right)}, \nonumber\\ &\forall (0,0) \neq (x_{k1,n1,1},x_{k2,n2,2}) \in \Delta\mathcal{S}_1 \times \Delta\mathcal{S}_2 , \nonumber\\ &\gamma > 1 , -\pi \leq \theta < \pi \rbrace
\label{ext_ci_eq1}
\end{align}
 The equation $|\gamma e^{j\theta} - c | \geq r $ is the exterior region of the circle centered at $c$ with radius $r$. Hence the result follows.  
\end{proof}
It can be verified that the Lemma 19 and Lemma 20 of \cite{muralidharan2013wireless} can be adapted for our case, with the modified definition of $c_{l_1,l_2}$ and $C_{l_1,l_2}$. Their statements are provided for completeness.
\begin{lemma}
The rings formed by $p_i$ , $1 \leq i \leq M_2/2$ are fully shaded.
\end{lemma}
Among the circles $c_{k_1,k_2}$, $k_1 \neq 2^{\delta}k_2$, $c_{M_1/2,1}$ is the outermost.
\begin{lemma}
The region between circles $c_{M_1/2,1}$ and $c_O$ is fully shaded.
\end{lemma}
\section*{Acknowledgment}
This work was supported partly by the Science and Engineering Research Board (SERB) of Department of Science and Technology (DST), Government of India, through J. C. Bose National Fellowship to B. Sundar Rajan.

\end{document}